\documentclass{article}

\usepackage{scrextend}
\usepackage[margin=1.2in]{geometry}
\usepackage{amsmath,amssymb,amsfonts,graphicx,amsthm,nicefrac,mathtools,bm}
\usepackage{hyperref}
\usepackage[numbers]{natbib}
\usepackage[english]{babel}
\usepackage{times}
\usepackage[T1]{fontenc}
\usepackage{bbm}
\usepackage{color,xcolor}
\usepackage{xspace}
\usepackage{rotating,multirow}

\newtheorem{theorem}{Theorem}

\newtheorem{proposition}{Proposition}

\theoremstyle{definition}

\theoremstyle{remark}

\makeatletter
\newtheorem*{rep@theorem}{\rep@title}
\newcommand{\newreptheorem}[2]{%
\newenvironment{rep#1}[1]{%
 \def\rep@title{#2 \ref{##1}}%
 \begin{rep@theorem}}%
 {\end{rep@theorem}}}
\makeatother
\newreptheorem{theorem}{Theorem}
\newreptheorem{lemma}{Lemma}
\newreptheorem{corollary}{Corollary}
\newreptheorem{proposition}{Proposition}

\newcommand{\eqnref}[1]{\eqref{eqn:#1}}

\def\independenT#1#2{\mathrel{\rlap{$#1#2$}\mkern2mu{#1#2}}}
\newcommand\independent{\protect\mathpalette{\protect\independenT}{\perp}}

\newcommand{\ignore}[1]{}

\newcommand{\bigxko}{\widetilde{X}}

\newcommand{\eqd}{\stackrel{\textnormal{d}}{=}}

\usepackage{tikz}
\usetikzlibrary{calc}

\author{Rina Foygel Barber\thanks{Department of Statistics, University of Chicago}\quad\quad  
Emmanuel J.~Cand{\`e}s\thanks{Departments of Mathematics and of Statistics, Stanford University} }
\title{On the Construction of Knockoffs in Case-Control Studies}

\date{December, 2018}

\begin{document}

\maketitle

\begin{abstract}

  Consider a case-control study in which we have a random sample,
  constructed in such a way that the proportion of cases in our sample
  is different from that in the general population---for instance,
the sample is constructed to achieve a fixed ratio of cases to
controls. Imagine that we wish to determine which of the potentially
many covariates under study truly influence the response by applying
the new model-X knockoffs approach.  This paper demonstrates that it
suffices to design knockoff variables using data that may have a
different ratio of cases to controls.  For example, the knockoff
variables can be constructed using the distribution of the original
variables under any of the following scenarios: (1) a population of
controls only; (2) a population of cases only; (3) a population of
cases and controls mixed in an arbitrary proportion (irrespective of
the fraction of cases in the sample at hand). The consequence is that
knockoff variables may be constructed using unlabeled data, which is
often available more easily than labeled data, while maintaining
Type-I error guarantees.

\end{abstract}

\section{Conditional Testing}\label{sec:intro}

In many scientific applications, researchers are often interested in
understanding which of the potentially many explanatory variables
truly influence a response variable of interest.  For example,
geneticists seek to understand the causes of a biologically complex
disease using single nucleotide polymorphisms (SNPs) as covariates. A
goal in such studies is to determine whether or not a given genetic
mutation influences the risk of the disease. Moving away from a
specific application, the general statistical problem is this: given
covariates $X_1,\dots,X_p$ and a response variable $Y$ which may be
discrete or continuous, for each variable $X_j$ we would like to know
whether the distribution of the response $Y| X_1, \ldots, X_p$
depends on $X_j$ or not; or equivalently, whether the $j$th variable
has any predictive power or not. Under mild conditions
\cite{DE:2000,candes2016}, this conditional null hypothesis is
equivalent to
\begin{equation}
  \label{eqn:null}
  H_j: \quad Y \independent X_j | X_{-j};
\end{equation}
under $H_j$, $X_j$ is independent of $Y$ once we have information
about all the other features.

It is intuitively clear that the null hypothesis of conditional
independence \eqnref{null} does not depend on the marginal
distribution of $X$. Specifically,~\eqnref{null} can be verified by simply checking that the conditional distribution of $Y| X$ 
depends on $X_{-j}$ and not on $X_j$---and therefore, knowing the conditional distribution of $Y| X$ is sufficient for testing
this property. Somewhat less intuitively, it is also the case that~\eqnref{null} can be verified through the conditional distribution
of $X| Y$, regardless of the marginal distribution of $Y$.
\begin{proposition}
\label{prop:null}
Consider any two distributions $P$ and $Q$ on the pair $(X,Y)$. Then:
\begin{itemize} 
\item Assume $P$ and $Q$ have the same likelihood of the response
    $Y$, i.e.~$P(Y|X) = Q(Y|X)$,\footnote{In this paper, we write
      joint distributions as $P(X,Y)$, marginals as $P(X)$ and $P(Y)$,
      and conditionals as $P(X|Y)$ and $P(Y|X)$.} and that $P(X)$ is
    absolutely continuous with respect to $Q(X)$.\footnote{The
      absolute continuity is here to avoid certain types of trivial
      situations of the following kind: take 
        $X \in \{0,1\}$ with $P(X = 0) = P(X = 1) = 1/2$ whereas
        $Q(X = 0) = 1$ and $Q(X = 1) = 0$. Since $X$ is constant
      under $Q$, $H_j$ holds trivially under $Q$. It may however not
      hold under $P$.} Then if $H_j$ is true under $Q$, it is also
    true under $P$. 
\item Assume $P$ and $Q$ have the same conditional distribution
  of the covariates, i.e.~$P(X|Y)=Q(X|Y)$, and that $P(Y)$ is
  absolutely continuous with respect to $Q(Y)$. Then if $H_j$ is true
  under $Q$, it is also true under $P$. Furthermore, in this case 
  we have
\begin{equation}
\label{eqn:conditional_null}
P(X_j | X_{-j}) = Q(X_j|X_{-j}). 
\end{equation} 
\end{itemize}
\end{proposition}
\begin{proof}[Proof of Proposition \ref{prop:null}]
  We prove the proposition in the case where all the variables are
  discrete; the case where some of the variables may be continuous is
  proved analogously. The first part of the proposition is nearly a
  tautology. Assume that $H_j$ is true under $Q$, then we
  have\footnote{To emphasize the role of absolute
    continuity, the equality below should be interpreted as holding at
    all points $(x,y)$ in the support of $P(X,Y)$.\label{abs-continuity}}
  \[
P(Y, X_j | X_{-j}) = P(Y|X) P(X_j|X_{-j}) = Q(Y|X) P(X_j| X_{-j}) = Q(Y|X_{-j})  P(X_j| X_{-j}). 
\]
The second inequality comes from the assumption that the likelihoods
are identical, and the third from our assumption that $H_j$ holds
under $Q$. Hence, $Y$ and $X_j$ are conditionally independent under
$P$, and so $H_j$ holds under $P$. 

For the second part, suppose that $H_j$ holds under $Q$. Then\footref{abs-continuity}
\[
P(X_j | Y, X_{-j}) = Q(X_j | Y, X_{-j}) = Q(X_j|X_{-j}),  
\]
where the first step holds because $P$ and $Q$ have the same
conditional distribution of $X| Y$, while the second step uses the assumption that $H_j$ holds
under $Q$, i.e.~$Y \independent X_j \, | \, X_{-j}$ under $Q$. This immediately
implies that $Y \independent X_j \, | \, X_{-j}$ under
$P$, and so $H_j$ holds under $P$. This gives
$P(X_j | Y, X_{-j}) = P(X_j | X_{-j}) = Q(X_j | X_{-j})$.
\end{proof}

\section{Case-Control Studies}\label{sec:case-control}

Prospective and case-control studies in which the response
$Y \in \{0,1\}$ takes on two values
\cite{Prentice1979}---e.g.~indicating whether an individual suffers
from a disease or not---offer well-known examples of distributions satisfying the second condition, where the conditional distribution
of $X| Y$ is the same but the marginal distribution of $Y$ is not.
\begin{description}
\item[Prospective study] In a prospective study, we may be interested
  in a specific population---all adults living in the UK, all males,
  all pregnant women, and so on. 

\item[Retrospective (case-control) study] In a retrospective study, individuals are
  typically recruited from the population based on the value of their
  response $Y$. In a case-control study, for instance, we may recruit
  individuals at random in such a way that the proportion of cases
  and controls achieves a fixed ratio. Typically, cases are are more
  prevalent in a retrospective sample than they are in a prospective
  sample. 
\end{description}

\noindent A prospective distribution $P$ and a retrospective distribution $Q$
have equal conditional distributions of $X| Y$,
\[
P(X | Y) = Q(X|Y).
\]
This is because, conditioning on $Y=1$ (the individual has the disease), both $P$ and $Q$ sample individuals
uniformly at random from the population of all individuals with the disease; the same holds for $Y=0$. That is, conditioned
on the value of $Y$, the two types of studies both sample $X$ from the same distribution. On the other hand,
$P$ and $Q$ will in general have different marginal distributions,
\[
P(X) \neq Q(X) \quad \text{and} \quad P(Y) \neq Q(Y).  
\]
For instance, while the incidence of a disease may be low (say, less
than 0.1\%) in the population, it may be high in the retrospective
sample (say, equal to 50\%). This trivially implies that
$P(Y)\neq Q(Y)$. In general we would also have $P(X)\neq Q(X)$ since,
under $Q$, values of $X$ associated with a high risk of the disease
would be overrepresented relative to $P$. Since $P(X| Y)=Q(X| Y)$,
however, it then follows from the second part of Proposition
\ref{prop:null} that in a case-control study, if conditional
independence holds w.r.t.~the retrospective distribution $Q$, it holds
w.r.t.~the prospective distribution $P$. (This is because the
retrospective distribution $Q$ includes both cases ($Y=1$) and
controls ($Y=0$) with positive probability and, therefore, $P(Y)$ is
absolutely continuous w.r.t.~$Q(Y)$.)

\section{Knockoffs in Case-Control Studies}

We now turn to the main subject of this paper. Model-X knockoffs is a
new framework for testing conditional hypotheses \eqnref{null} in
complex models. While most of the literature relies on a specification
of the model that links together the response and the covariates, the
originality of the knockoffs approach is that it does not make any
assumption about the distribution of $Y | X$. The price to pay for
this generality is that we would need to know the marginal
distribution of the covariates.  Assume we get independent samples
from a distribution $Q(X,Y)$ (as in a retrospective study, for
example). Model-X knockoffs are fake variables
$\tilde X_1, \ldots, \tilde X_p$ obeying the following pairwise
exchangeability property:
\begin{equation}\label{eqn:swap_j_only}
  X \sim Q(X) \quad \implies \quad   \big(X_j,\bigxko_j,X_{-j},\bigxko_{-j}\big) 
  \eqd \big(\bigxko_j,X_j,X_{-j},\bigxko_{-j}\big) \quad \text{for any
  } j \in \mathcal{H}_0.
\end{equation}
Here, $\mathcal{H}_0\subset\{1,\dots,p\}$ is the subset of null
hypotheses that are true, i.e.~covariates $j$ for which
$Y\independent X_j \, | \, X_{-j}$ under $Q$ (and, therefore,
$Y\independent X_j \, | \, X_{-j}$ hold also under any other
distribution with the same conditional). Having
achieved~\eqnref{swap_j_only}, a general selection procedure
effectively using knockoff variables as negative controls can be
invoked to select promising variables while rigorously controlling the
false discovery rate. In other words,~\eqnref{swap_j_only} implies
that a variable selection procedure that is likely to mistakenly
select irrelevant variable $X_j$, is equally likely to select the
constructed knockoff feature $\bigxko_j$, which then alerts us to the
fact that our variable selection procedure is selecting false
positives. We refer the reader to the already extensive literature on
the subject, e.g.~\cite{barber2015,candes2016}, for further
information.

In the literature, we often encounters the claim that this {\em shift
  in the burden of knowledge}---i.e.~knowledge about the distribution
of $X$ versus that of $Y | X$---is appropriate in situations where we
may have ample unlabeled data available to `learn' the distribution of
the covariates $X$. After all, while the geneticist may have observed
only a few instances of a rare disease, she may have at her disposal
several hundreds of thousands of unlabeled genotypes.  This means that
we have very limited access to labeled data, i.e.~pairs $(X,Y)$, where
$Y$ is known and where the sample is balanced to have a non-vanishing
proportion of cases (i.e.~$Y=1$)---this is the retrospective
distribution $Q$.  In contrast, unlabeled data ($X$ only) is easy to
obtain, but will be drawn from the general population, in which $Y=1$
is extremely rare---that is, drawn from the prospective distribution
$P$.  Imagine using this unlabeled data to learn the prospective
distribution $P(X)$, i.e.~the distribution of $X$ in the general
population, and then using this knowledge for variable selection using
our labeled case-control data, i.e.~draws from the retrospective
distribution $Q(X,Y)$.  Using the distribution $P(X)$ learned on the
unlabeled data, we would in principle be able to construct
exchangeable features for $P(X)$, i.e.~knockoff variables
  $\bigxko_1,\dots,\bigxko_p$ constructed to satisfy the
  exchangeability property
\begin{equation}\label{eqn:swap_j_onlyP}
  X \sim P(X) \quad \implies \quad   \big(X_j,\bigxko_j,X_{-j},\bigxko_{-j}\big) 
  \eqd \big(\bigxko_j,X_j,X_{-j},\bigxko_{-j}\big)  \quad \text{for any
  } j \in \mathcal{H}_0. 
\end{equation}
Now contrast \eqnref{swap_j_only} and \eqnref{swap_j_onlyP}:
we want exchangeability w.r.t.~the
retrospective distribution $Q$, but since we have constructed our knockoffs using the unlabeled
data, we have
perhaps only achieved exchangeability w.r.t.~the prospective distribution
$P$. The good news is that this mismatch {\em does not affect the validity of our inference}.
By Proposition~\ref{prop:null}, exchangeability of the null features and their knockoffs
under the prospective distribution implies exchangeability under the
retrospective distribution. A more general statement is this:

\begin{theorem}
  \label{thm:main}
  Consider two distributions $P$ and $Q$ such that
  $P(X_j| X_{-j}) = Q(X_j | X_{-j})$ for every null variable, i.e.~for all $j \in {\cal H}_0$.
  Then any knockoff sampling scheme obeying exchangeability w.r.t.~$P$
  \eqnref{swap_j_onlyP} obeys the same property w.r.t.~$Q$
  \eqnref{swap_j_only}.

  By~\eqnref{conditional_null} of Proposition \ref{prop:null}, this
  conclusion applies to any situation where $P$ and $Q$ have the same
  conditionals, i.e.~$P(X|Y) = Q(X|Y)$ (with the proviso that $P(Y)$
  is absolutely continuous w.r.t.~$Q(Y)$). In particular, it applies to 
  case-control studies in which $Q$ is a
    retrospective distribution and $P$ is either a population of
    controls only, or a population of cases only, or a population of
    cases and controls mixed in an arbitrary proportion (irrespective
    of the fraction of cases  in the sample drawn from $Q$).
   \end{theorem}

   This result allows considerable flexibility in the way we can
   construct knockoff variables since we can use lots of unlabeled
   data to estimate conditional distributions $X_j | X_{-j}$.
For example, by constructing our knockoffs from a data set consisting
of controls only, which does not match the population in a
case-control study, we are nonetheless using the correct conditionals
$X_j|  X_{-j}$ for every null variable $j$ and can be assured that
we are constructing valid knockoffs.

\begin{proof}[Proof of Theorem \ref{thm:main}]
  Once again, we prove the result in the case where all the variables
  are discrete. To prove our claim, we need to show the following:
  when $X \sim Q(X)$, the distribution of $X_j, \bigxko | X_{-j}$
  is symmetric in the variables $X_j$ and $\bigxko_{j}$.  This
  distribution is given by
\[
Q(X_j|X_{-j}) P(\bigxko | X) = P(X_j | X_{-j}) P(\bigxko
| X), 
\]
where $P(\tilde X | X)$ denote the conditional distribution of
$\bigxko | X$, and the equality holds since $Q(X_j|X_{-j}) =P(X_j|X_{-j})$ by assumption. 
 Our claim now follows from \eqnref{swap_j_onlyP},
the exchangeability of knockoffs and null variables under $P$,
which implies that the right-hand side is symmetric in $X_j$ and
$\bigxko_j$. Therefore, $X_j$ and $\bigxko_j$ are also exchangeable under $Q$, proving
the theorem.
\end{proof}

\section{Discussion}

Our main result shows that if we use the right conditionals
$X_j | X_{-j}$ for each null variable, then the model-X framework
applies and, ultimately, inference is valid---even when we construct
knockoffs with reference to a distribution with the wrong marginals
$P(X)$ and $P(Y)$.  Mathematically, this result can be deduced from
the arguments in \cite{barber2018robust}.  Our contribution here is to
link this phenomenon with the situation in case-control studies as
specialists have openly wondered about the validity of knockoffs
methods in such settings \cite{marchini2019comment}. Not only is the
approach valid but we can further leverage the shift in the burden of
knowledge, using the ample availability of unlabeled data to construct
valid knockoffs.

We have not discussed the question of power in this brief
paper. However, we pose an interesting question for further
investigation: now that we know that we can use either a population of
controls to construct knockoffs, or a population of cases, or a
population in which cases and controls are in an arbitrary proportion,
which population should we use as to maximize power? We hope to report
on this in a future paper.

\subsection*{Acknowledgements}
R.~F.~B.~was partially supported by the National Science Foundation
via grant DMS 1654076, and by an Alfred P.~Sloan fellowship.
E.~C.~was partially supported by the Office of Naval Research under
grant N00014-16-1-2712, by the National Science Foundation via DMS
1712800, by the Math + X Award from the Simons Foundation and by a
generous gift from TwoSigma. E.~C.~would like to thank Chiara Sabatti
and Eugene Katsevich for useful conversations related to this
project.

\bibliographystyle{plainnat}
\bibliography{bibliography}
\end{document}